\documentclass[a4paper,10pt,preprintnumbers,floatfix,superscriptaddress,aps,prl,twocolumn,showpacs,notitlepage,longbibliography,nofootinbib]{revtex4-2}

\usepackage{amsmath, amssymb, amsthm, thm-restate, bm, physics, dsfont, nicefrac}
\usepackage{xcolor, graphicx, subfigure, float, appendix, booktabs, microtype, afterpage}
\definecolor{pgreen}{RGB}{84, 129, 102}
\definecolor{porange}{RGB}{199, 103, 42}
\usepackage[colorlinks,citecolor=pgreen,linkcolor=porange,urlcolor=pgreen]{hyperref}

\renewcommand{\paragraph}[1]{\addcontentsline{toc}{section}{#1}\emph{#1.}---}
\newcommand{\sn}[1]{sn\left( #1 \right)}
\newcommand{\eye}{\mathbf{1}}
\newcommand{\hilb}{H}
\renewcommand{\tr}[1]{\text{tr}\left( #1 \right)}
\newcommand{\ptr}[2]{\text{tr}_{#1} \left( #2 \right)}
\newcommand{\ptrb}[2]{\text{tr}_{#1} \left[ #2 \right]}

\DeclareMathOperator{\tmin}{\otimes_{\text{min}}}
\DeclareMathOperator{\tmax}{\otimes_{\text{max}}}

\begin{document}
\title{Complete hierarchy for high-dimensional steering certification}

\author{Carlos de Gois}
\affiliation{Naturwissenschaftlich-Technische Fakult\"{a}t, Universit\"{a}t Siegen, Walter-Flex-Stra\ss e 3, 57068 Siegen, Germany}

\author{Martin Pl\'{a}vala}
\affiliation{Naturwissenschaftlich-Technische Fakult\"{a}t, Universit\"{a}t Siegen, Walter-Flex-Stra\ss e 3, 57068 Siegen, Germany}

\author{Ren\'{e} Schwonnek}
\affiliation{Institut f\"{u}r Theoretische Physik, Leibniz Universit\"{a}t Hannover, Appelstra\ss e 2, 30167, Hannover, Germany}

\author{Otfried G\"{u}hne}
\affiliation{Naturwissenschaftlich-Technische Fakult\"{a}t, Universit\"{a}t Siegen, Walter-Flex-Stra\ss e 3, 57068 Siegen, Germany}

\begin{abstract}
    High-dimensional quantum steering can be seen as a test for the dimensionality of entanglement, 
    where the devices at one side are not characterized. As such, it is an important component in 
    quantum informational protocols that make use of high-dimensional entanglement. Although it has 
    been recently observed experimentally, the phenomenon of high-dimensional steering is lacking 
    a general certification procedure. We provide necessary and sufficient conditions to 
    certify the entanglement dimension in a steering scenario. These conditions are stated in terms 
    of a hierarchy of semidefinite programs, which can also be used to quantify the phenomenon using 
    the steering dimension robustness. To demonstrate the practical viability of our method, we 
    characterize the dimensionality of entanglement in steering scenarios prepared with maximally 
    entangled states measured in mutually unbiased bases. Our methods give significantly stronger bounds on the noise robustness necessary to experimentally certify high-dimensional entanglement.
\end{abstract}

\maketitle

\paragraph{Introduction}
Controlling increasingly higher-dimensional quantum systems is one of the keys that can 
unlock the advantage of quantum technologies over classical predecessors. Indeed, a central 
promise of quantum computing is an improved scaling of the required qubits in comparison to 
classical bits, with a similar perspective applying, e.g., for quantum metrology or the capacity 
of quantum communication channels. On the other side, employing high-dimensional systems can improve
the noise robustness of quantum information protocols and experiments \cite{vertesi-qudit-nonlocality,aubrun2022monogamy,ecker-entanglement-distribution,
zhu-high-dimensional-photonic-entanglement,skrzypczyk2015loss,marciniak2015unbounded,qu2022retrieving,miklin2022exponentially}, and by this the central bottleneck for applications like quantum cryptography 
or a quantum network can be removed.

In any of the above applications, the genuine use of high-dimensionality demands the ability 
of creating and controlling entangled quantum states with a high Schmidt number. In operational 
terms, this number asks for the minimal local Hilbert space dimension $k$ that two parties, 
Alice and Bob, have to possess for holding shares of an entangled state $\rho_{AB}$.
    
Certifying this number in the context of the one-side device-independent setting of a 
steering experiment [see Fig.~\ref{fig:steering}] is a task that recently gained a significant amount 
of attention \cite{designolle-highdim-steering,designolle2022robust,qu2022robust,jones2022equivalence}. 
Here we assume a situation in which only one party, say Bob, has the ability to fully characterise 
his local quantum system. The specifics of the other party, Alice, are kept hidden. The only way 
in which Alice can interact with her system is by applying black-box measurements, of which she 
can only control an input $x$ and observe a corresponding output $a$.  
    
Alice's task in this situation is to find suitable black-box measurements and an initial shared
quantum state $\rho$ that allows her to convince Bob that they hold a state with a 
Schmidt number at least $k$. This task can be seen as a fundamental building block 
for the verification of quantum hardware, since a successful Alice will in the same 
run also prove her ability to control and non-trivially manipulate quantum states on 
a $k$ dimensional Hilbert space. 

    \begin{figure}\label{fig:intro}
        \centering
        \subfigure[Steering experiment]{\label{fig:steering}\includegraphics[width=.9\columnwidth]{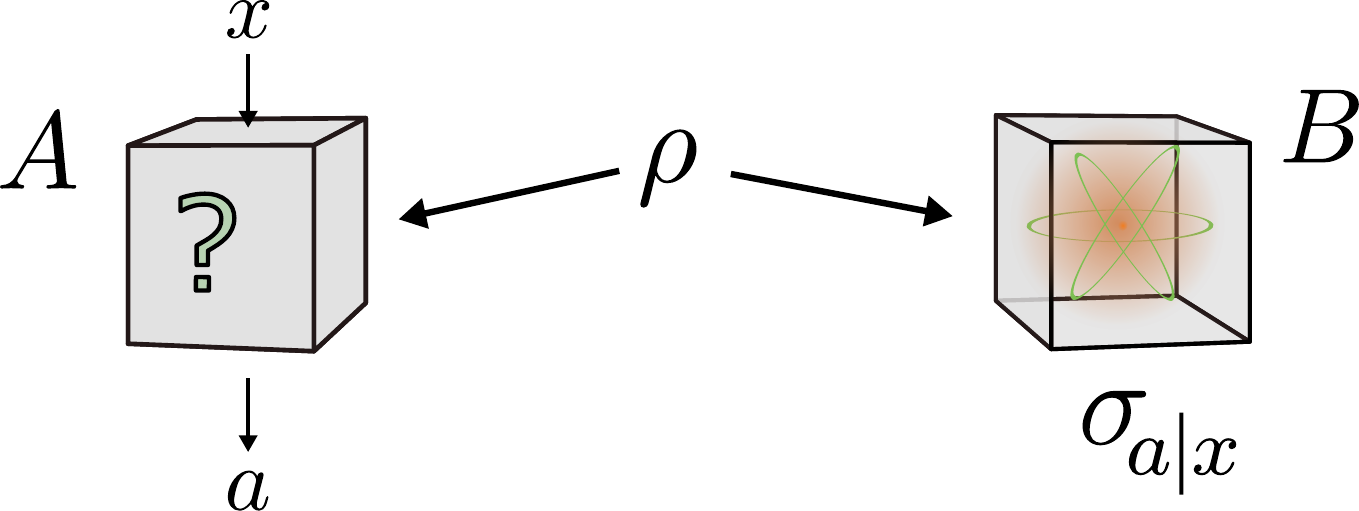}}
        
        \subfigure[$k$-preparable assemblages]{\label{fig:k-preparable-sets}\includegraphics[width=.7\columnwidth]{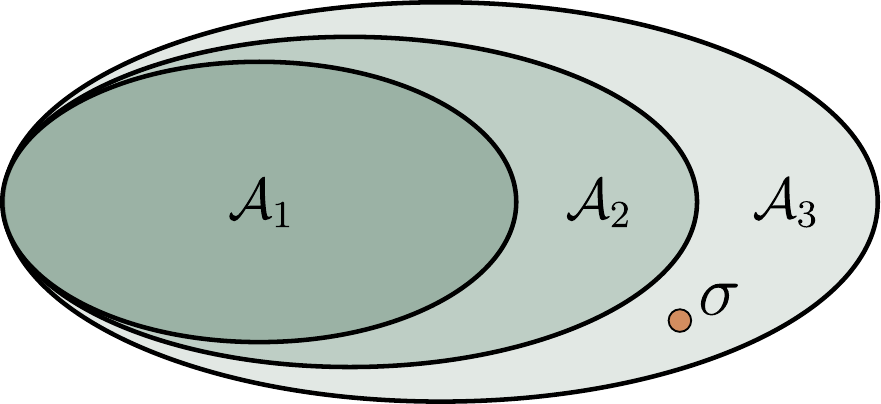}}
        \caption{In a steering experiment, Alice performs local measurements on a bipartite quantum state 
        $\rho$ shared between her and Bob. This results in a set of conditional states on Bob's side [a 
        so-called state assemblage $\sigma$, see Eq.~\eqref{eq:assemblage-elements}], which can only be 
        steerable if the shared state $\rho$ were entangled. The central task in high-dimensional steering 
        is to use an assemblage not only to witness entanglement in $\rho$, but to also quantify it. We say 
        an assemblage is $k$-preparable assemblage (and write that $\sigma \in \mathcal{A}_k$) if it 
        can be prepared from a state $\rho$ with Schmidt number smaller or equal to $k$.}
    \end{figure}

    
    Despite its fundamental importance, the certification of high-dimensional steering can currently only 
    be done for a very restricted set of cases. The usual tool employed here are linear steering witnesses, 
    which typically come with the drawback that they are always designed around a given set of scenario 
    and resources that one aims to characterize. They usually only perform well when the corresponding 
    black-box setting is applied, but tend to fail otherwise.  For example, albeit in 
    Ref.~\cite{designolle-highdim-steering} tight witnesses are derived, they are limited to pairs of 
    measurements. Their methods were later extended to more settings, but the obtained witnesses 
    are quite weak \cite{designolle2022robust,qu2022robust}. 
    
    In this paper, we provide a general and complete hierarchy of semidefinite programs for high-dimensional steering certification. The hierarchy holds for any number of measurement outcomes and measurement settings and it is always complete, meaning that for a given state assemblage on Bob's side one can in principle exactly determine the dimension of the entanglement that the assemblage certifies. We use the hierarchy to numerically obtain significantly improved bounds on the noise robustness of assemblages that certify entanglement dimension $k > 2$.
    
\paragraph{High-dimensional steering}
    In a steering experiment, Alice attempts to ``steer'' Bob's system by performing local measurements on her share of a quantum state $\rho$. Her measurements are uncharacterised, and her only input to the experiment is a label $x \in \{1, \ldots, X\}$, indicating a measurement choice. We denote her measurements by $\mathcal{M}_x = \{ M_{a \vert x} \}$, where $M_{a \vert x} \geq 0$ are the effects and $\sum_a M_{a \vert x} = \eye$.
    After Alice measures and announces her outcome to Bob, Bob's state ends up being
    \begin{equation}
        \sigma_{a \vert x} = \text{tr}_A \left[ \rho \left( M_{a \vert x} \otimes \eye_B \right) \right] .
        \label{eq:assemblage-elements}
    \end{equation}
    The collection of subnormalised states $\sigma = \{ \sigma_{a \vert x} \}$ is called an assemblage. Since Alice cannot communicate with Bob, we require the assemblages to be nonsignalling: $\sum_a \sigma_{a \vert x} = \sum_a \sigma_{a \vert x^\prime}$, for any choice of measurements $x$ and $x^\prime$. An assemblage is said to demonstrate steering if it cannot be explained by means of a local hidden state (\textsc{lhs}) model \cite{wiseman-steering}
    \begin{equation}
        \sigma_{a \vert x} = \int_\Lambda d\lambda p(\lambda) p(a \vert x, \lambda) \sigma_\lambda ,
        \label{eq:lhs-model}
    \end{equation}
    where the states $\sigma_\lambda \in \mathcal{B}(\hilb_B)$ are local to Bob, and $\lambda$ is a latent, classical variable correlating Alice's and Bob's devices. Whenever the shared state $\rho$ is separable, that is, whenever it can be written as $\rho = \sum_\lambda p(\lambda) \rho^A_\lambda \otimes \rho^B_\lambda$ for some local states $\rho^A_\lambda$ and $\rho^B_\lambda$, all assemblages $\sigma$ prepared from $\rho$ admit an \textsc{lhs} model. Therefore, any steerable $\sigma$ certifies entanglement of the shared state.
    
    Entanglement, though, comes in many forms. In particular, it can be quantified \cite{vedral-quantifying-entanglement}, but steerability does not provide insights into how entangled $\rho$ is. An object of recent discussion is whether, and how, can we use steering experiments to not only certify, but also quantify entanglement.     The Schmidt number \cite{terhal-schmidt-number} (also called ``entanglement dimension'') is a popular quantifier in this and other correlation scenarios, but other concepts exist \cite{kraft2018characterizing}.
    
    Recall that a pure bipartite state $\ket{\psi}$ has Schmidt rank $k$ if its Schmidt decomposition has $k$ terms, $\ket{\psi} = \sum_{i=1}^k \nu_i \ket{i,i}$. Extending this definition to mixed states, we say that a state $\rho$ has Schmidt number $\sn{\rho} = k$ if (i) for any decomposition $\rho = \sum_i p_i \dyad{\psi_i}{\psi_i}$, at least one of the vectors $\ket{\psi_i}$ has Schmidt rank at least $k$, and (ii) there exists a decomposition of $\rho$ with all vectors $\ket{\psi_i}$ having Schmidt rank at most $k$. We define the sets $\mathcal{S}_k = \{ \rho \mid \sn{\rho} \leq k \}$, and notice that they are convex, their extremal points are pure states, and $\mathcal{S}_{k - 1} \subset \mathcal{S}_{k}$ for all $k \geq 2$, where $\mathcal{S}_1$ is the set of separable states.
    
    Building on top of it, we say that an assemblage $\sigma$ is $k$-preparable if it can be prepared by local measurements on a state $\rho \in \mathcal{S}_k$ and we use the notation $\sn{\sigma} \leq k$ to mean so, i.e., $\sn{\sigma} \leq k$ if and only if $\sigma_{a \vert x} = \text{tr}_A \left[ \rho \left( M_{a \vert x} \otimes \eye_B \right) \right]$ for some $\rho \in \mathcal{S}_k$. Furthermore, we define $\mathcal{A}_k = \{ \sigma \mid \sn{\sigma} \leq k \}$ as the set of all assemblages that can be prepared with states in $\mathcal{S}_k$. The sets $\mathcal{A}_k$ exhibit a nested structure $\mathcal{A}_{k - 1} \subset \mathcal{A}_k$, for any $k \geq 2$ \cite{designolle-highdim-steering}.
    
    Naturally, the central question in high-dimensional steering is how to characterise these sets. In practice, given an assemblage, we want to be able to certify that it is not $k$-preparable.

\paragraph{Main result}
    Suppose $\sigma = \{\sigma_{a \vert x}\}$ is an assemblage in $\mathcal{A}_k$. Then, it can be obtained by means of a state $\rho$ with Schmidt number $k$. Any such state can be seen as coming from a separable operator in an extended space, where an entangling projection was made. More precisely, we extend Alice's ($A$) and Bob's ($B$) subsystems with $k$-dimensional auxiliary spaces $A^\prime$ and $B^\prime$ and define the projection $\Pi_k = \eye \otimes \sum_{i=1}^k \ket{i,i}_{A^\prime B^\prime}$. We call $\Pi_k$ an entangling projection since it can be seen as an unnormalised maximally entangled (\textsc{me}) state $\ket{\phi^+_d} = \nicefrac{1}{\sqrt{d}}\sum_{i=1}^{d} \ket{i,i}$. Then, any pure state $\ket{\psi} = \sum_{i=1}^k \eta_i \ket{i_A, i_B}$ can be expressed as $\ket{\psi} = \Pi^\dag_k \left( \ket{i^\prime_{AA^\prime}, i^\prime_{B^\prime B}} \right)$ if we set
    \begin{equation}
        \ket{i^\prime_{AA^\prime}} = \sum_{i=1}^k \ket{i_A, i_{A^\prime}},\,\, \ket{i^\prime_{B^\prime B}} = \sum_{i=1}^k \eta_i \ket{i_{B^\prime}, i_B} .
    \end{equation}
    Consequently, for mixed states,
    \begin{equation}
        \rho = \Pi^\dag_k \big( \underbrace{ \sum_i p_i \rho_{AA^\prime}^i \otimes \rho_{B^\prime B}^i }_{\Omega} \big) \Pi_k .
        \label{eq:lifting-down}
    \end{equation}

    Therefore, any state $\rho$ with Schmidt number $k$ can be associated with an operator $\Omega$ which is separable w.r.t.\@ the $AA^\prime/B^\prime B$ partition, at the expense of an entangling projection $\Pi_k$ between $A^\prime$ and $B^\prime$ (Fig.\@ \ref{fig:lifting}). This fact has already been noted \cite{hulpke-liftings} and used to derive a Schmidt number certification hierarchy for entanglement \cite{weilenmann-faithful-entanglement}.
    
    \begin{figure}
        \centering
        \includegraphics[width=.95\columnwidth]{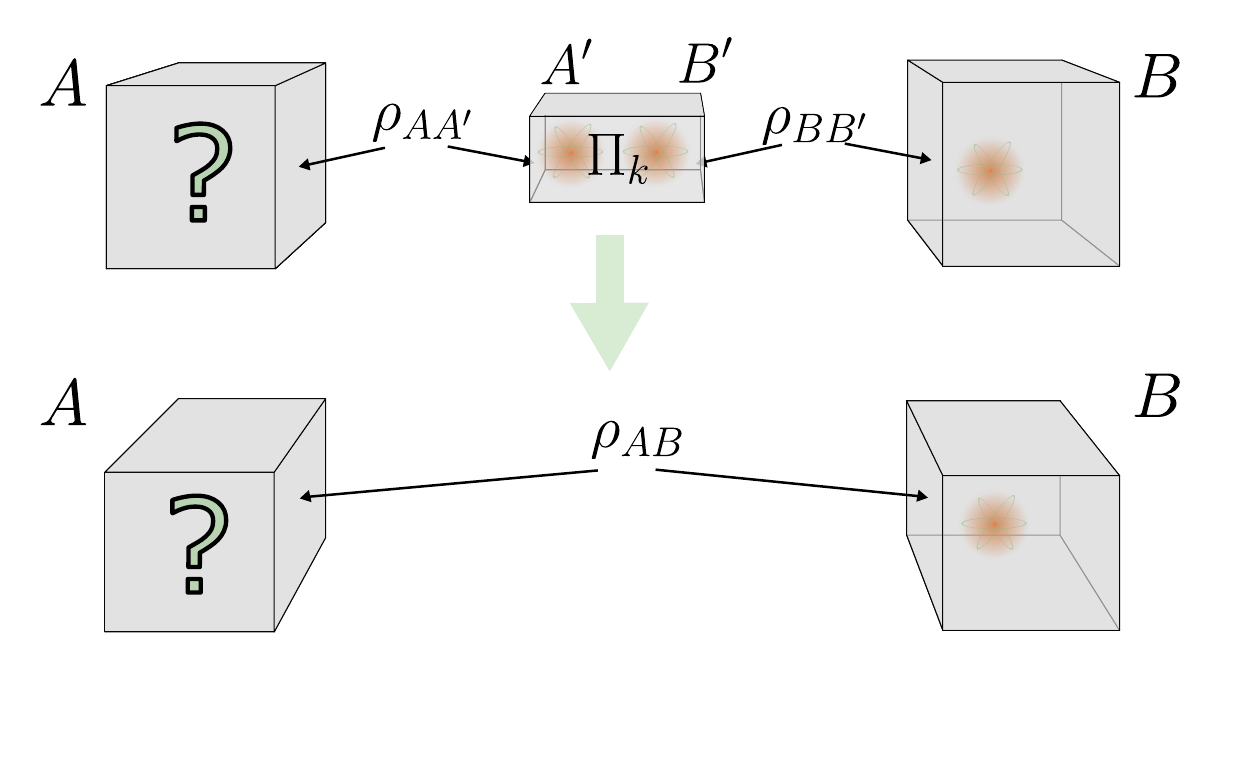}
        \caption{Any state $\rho$ with Schmidt number $k$ acting on $\mathcal{B}(\hilb_A \otimes \hilb_{B})$ is equivalent to some separable operator $\sum_i p_i \rho_{AA^\prime}^i \otimes \rho_{B^\prime B}^i$ in an extended Hilbert space, which is then projected with $\Pi_k = \eye \otimes \sum_{i=1}^k \ket{i,i}_{A^\prime B^\prime}$ according to Eq.~\eqref{eq:lifting-down}.}
        \label{fig:lifting}
    \end{figure}

    With this in mind, when $\rho$ has Schmidt number $k$ we can recast the r.h.s.\@ of Eq.~\eqref{eq:assemblage-elements} as
    \begin{equation}
        \begin{split}
            \Pi^\dag_k& \left[ \sum_i p_i \text{tr}_{A} \left( \rho_{AA^\prime}^i \left( M_{a \vert x} \otimes \eye_{A^\prime} \right)  \right) \otimes \rho_{B^\prime B}^i \right] \Pi_k \\
            &= \Pi_k^\dag \bigg[ \underbrace{ \sum_i p_i \tau_{A^\prime}^{i,a,x} \otimes \rho_{B^\prime B}^i}_{\Omega_{a \vert x}} \bigg] \Pi_k .\label{eq:omega-operators}
        \end{split}
    \end{equation}
    Here, the set of $\tau_{A^\prime}^{i,a,x}$ operators has the structure of an assemblage and in particular the nonsignalling condition $\sum_{i,a} \tau_{A^\prime}^{i,a,x} = \sum_{i,a} \tau_{A^\prime}^{i,a,x^\prime}, \,\forall x \neq x^\prime$ holds.
    
    By this construction, an assemblage $\{ \sigma_{a \vert x} \}$ is $k$-preparable if and only if there are operators $\Omega_{a \vert x} \in \mathcal{B}\left( \hilb_k \otimes \hilb_k \otimes \hilb_B \right)$ such that $\sigma_{a \vert x} = \Pi_k^\dag \Omega_{a \vert x} \Pi_k$, where $\Omega_{a \vert x} = \sum_i p_i \tau_{A^\prime}^{i,a,x} \otimes \rho_{B^\prime B}^i$ and $\{ \tau_{A^\prime}^{i,a,x} \}$ is nonsignalling.
    
    Although this provides an insight into the structure of high-dimensional steering assemblages, it does not evidence a way of determining the existence of the $\Omega_{a \vert x}$. As it turns out, this can be solved by means of a complete hierarchy of semidefinite tests, which can be seen as a generalisation of the symmetric extensions criterion for entanglement certification \cite{dps-hierarchy}.
    
    Given $N \geq 1$, $\Omega_{a \vert x}$ is said to have a nonsignalling symmetric extension of order $N$ if there exists an operator $\Xi_{a \mid x} \in \mathcal{B}\left( \hilb_{A^\prime} \otimes \hilb_{B^\prime B}^{\otimes N} \right)$ such that
    \begin{subequations}\label{eq:ns-symmetric-extension}
    \begin{align}
        &\quad \ptr{(B^\prime B)_i \,:\, i \neq m}{\Xi_{a \mid x}} = \Omega_{a \vert x}, \;\forall m \in \{1, \ldots, N\}, \\
        &\quad \sum_a \Xi_{a \mid x} = \sum_a \Xi_{a \mid x^\prime} ,\,\forall x, x^\prime .
    \end{align}
    \end{subequations}
    Here, the first condition ensures that when all but one of the systems in $B^\prime B$ are traced out, $\Omega_{a \vert x}$ is recovered, while the second condition enforces nonsignalling in the extensions.
    
    With the following theorem, we link the existence of symmetric extensions with that of the $\Omega_{a \vert x}$ operators described in Eq.~\eqref{eq:omega-operators}.

    \begin{restatable}{thm}{symmetricextensions}
    \label{thm:symmetric-extensions}
        An assemblage $\{ \sigma_{a \vert x} \}$ is $k$-preparable if and only if there are corresponding operators $\Omega_{a \vert x} \in \mathcal{B} (\hilb_{A^\prime} \otimes \hilb_{B^\prime B} )$ such that $H_{A^\prime}$, $H_{B^\prime}$ are of dimension $k$ and  $\Omega_{a \vert x}$ has a nonsignalling symmetric extension [Eq.~\eqref{eq:ns-symmetric-extension}] to any order $N$.
    \end{restatable}
    \begin{proof}
    The ``if'' condition, essential for the applications described ahead, is easily seen to be satisfied by taking $\Xi_{a \mid x} = \sum_i p_i \tau_{A^\prime}^{i,a,x} \otimes \left( \rho_{B^\prime B}^i \right)^{\otimes N}$, where we, without loss of generality, assumed that $\tr{\rho_{B^\prime B}^i} = 1$. We postpone the other direction of the proof to the Appendix.
    \end{proof}

    Theorem \ref{thm:symmetric-extensions} tells us that, to test if an assemblage $\{ \sigma_{a \vert x} \}$ is not $k$-preparable, we can check whether there exists a nonsignalling symmetric extension of $\Omega_{a \vert x}$ to some order $N$. Thus, for some given assemblage $\left\{ \sigma_{a \vert x} \right\}$, Schmidt number $k$ and hierarchy level $N$, we must search for extensions $\left\{ \Xi_{a \vert x} \right\}$ that satisfy the constraints in Eq.~\eqref{eq:ns-symmetric-extension} with $\sigma_{a \vert x} = \Pi_k^\dag \Omega_{a\vert x} \Pi_k$.
    
    This can be done with semidefinite programming. If the program is unfeasible (i.e., the symmetric extension does not exist), then $\{ \sigma_{a \vert x} \}$ is not $k$-preparable. Otherwise, the test is inconclusive and we can proceed to test a higher $N$. This hierarchy is complete in the sense that unfeasibility will eventually occur for all assemblages which are not $k$-preparable, but never for $k$-preparable ones.

\paragraph{Certifying high-dimensional steering with maximally entangled states and \textsc{mub} measurements}
    Consider, as an example, the assemblage arising from a set of $X$ $d$-dimensional mutually unbiased bases (MUB) measurements \cite{durt-mubs} acting on the maximally entangled state $\ket{\phi^+_d} = \nicefrac{1}{\sqrt{d}}\sum_{i=1}^{d} \ket{i,i}$,
    \begin{equation}
        \sigma = \Big\{ \left\{ \ptrb{A}{\dyad{\phi^+_d} \cdot \left( M_{a \vert x} \otimes \eye_B \right)} \right\}_{a = 1}^{d} \Big\}_{x = 1}^{X} .
    \end{equation}
    By mixing $\sigma$ with the white noise assemblage $\sigma_\eye$ (the one with elements $\nicefrac{\eye}{d^2}$) at a visibility $\eta$, we get a new assemblage with elements $\sigma_{a \vert x}(\eta) = \eta \sigma_{a \vert x} + (1 - \eta) {\eye}/{d^2}$. Since $\sigma_\eye$ is always a feasible point in the program described above, one may rewrite it as a maximisation on $\eta$. Any solution $\eta^* < 1$ to this maximisation problem certifies that the assemblage $\{ \sigma_{a \vert x} \}$ does not have nonsignalling symmetric extensions and therefore that $\sn{\sigma} > k$.
    
    However, even for reasonable values of the putative Schmidt number $k$ and hierarchy level $N$, the size of the program increases superexponentially. To mitigate this blow up, one can make use of the symmetries in the problem.
    
    First, observe that if the extension is of the form $\Xi_{a \mid x} = \sum_i p_i \tau_{A^\prime}^{i,a,x} \otimes \left( \rho_{B^\prime B}^i \right)^{\otimes N}$ then we can, without loss of generality, assume that $\rho_{B^\prime B}^i$ are pure states and thus
    $\Xi_{a \vert x}$ acts on $S = \hilb_{A^\prime} \otimes \text{Sym}_N( B^\prime B )$, where $\text{Sym}_N( B^\prime B )$ is the symmetric subspace of $\mathcal{B}( \hilb_{B^\prime B}^{\otimes N} )$. This space is of dimension $d_S = d_{A^\prime} \times \binom{k d_B + N - 1}{k d_B - 1}$ \cite{watrous}, considerably reducing the number of variables in our problem. To put this observation to use, notice that, from the projector
    \begin{equation}
        \Pi_S = \eye_{A^\prime} \otimes \frac{1}{n!} \sum_{\pi \in S_n} P_\pi ,
    \end{equation}
    where $S_N$ is the $N$-elements permutation group, one may construct a matrix $P_S$ whose $d_S$ rows span the symmetric subspace. With that in hand, we can start the whole ordeal with a $d_S \times d_S$ matrix ${\Xi}^\prime_{a \vert x}$ and substitute $\Xi_{a \vert x}$ for $P_S^\dag \Xi^\prime_{a \mid x} P_S$ in all constraints. More than that, due to the permutational symmetry, all partial traces need only be evaluated on one of the $B^\prime B$ subsystems. Therefore, we not only reduce the dimension on the optimisation variables, but also decrease the number of constraints.
    
    Figure \ref{fig:results} shows numerical results obtained with an implementation of the method hereby described. We additionally made use of positive partial transpose constraints --- which are known to accelerate convergence in the symmetric extensions hierarchy \cite{navascues2009power} --- and explicitly enforced that $\tr{ \Xi_{a \mid x} } = k \tr{\sigma_{a \vert x}}$, which also leads to better results.
    
    \begin{figure}
        \centering
        \includegraphics[width=.99\columnwidth]{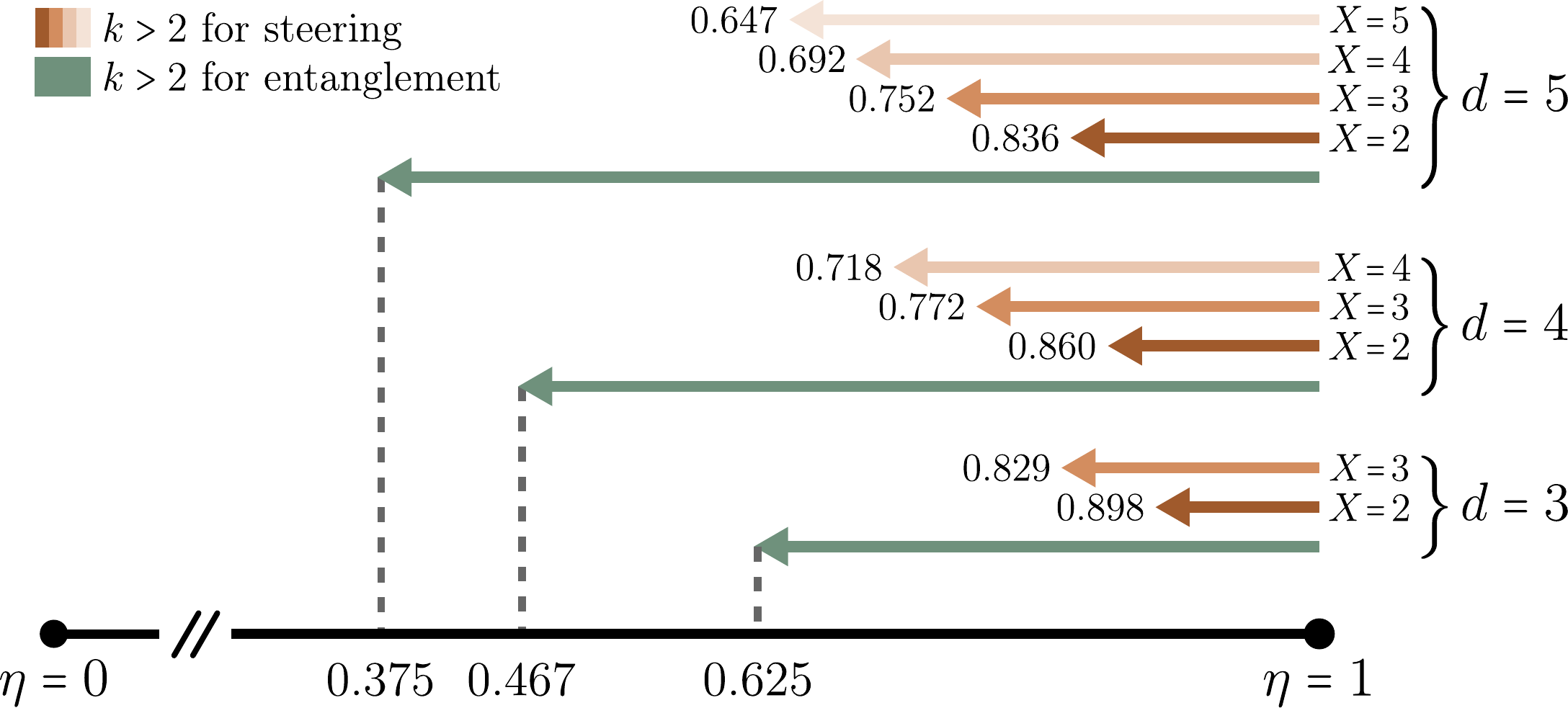}
        \caption{Schmidt number ($k$) certification for assemblages prepared with maximally entangled states and $X$ measurements on \textsc{mub}s. The assemblages were mixed with the white noise assemblage and an upper bound on the minimum visibility such that $2$-preparability can be falsified was computed. All results were computed with two copies of the $B^\prime B$ system (i.e., $N=2$). Increasing $N$ can lead to smaller bounds. For instance, in the $d=3, X=3$ case, by setting $N=3$ we were able to certify that $k > 2$ for all $\eta \geq 0.789$.
        For comparison, we also show analogous bounds for the entanglement scenario, obtained through a similar procedure: To certify the Schmidt number in an entanglement scenario, the standard symmetric extensions hierarchy can be directly applied to the observation illustrated in Fig.\@ \ref{fig:lifting} (See also \cite{weilenmann-faithful-entanglement}).}
        \label{fig:results}
    \end{figure}
    
    Even though the results shown in Fig.\@ \ref{fig:results} were computed for the lowest hierarchy level ($N=2$), they provide better bounds than the known values for high-dimensional steering witnesses with $X>2$ measurements \cite[Table III]{designolle2022robust}, except for the $d=5$, $X=3$ case. In comparison to the previous results for $X=2$ measurements \cite{designolle-highdim-steering}, level $N=2$ of our hierarchy provides slightly worse visibilities, but moving to level $N=3$ makes up for it. For example, while they find a visibility of $0.886$ ($0.838$) for dimension $3$ ($4$), our method improves it to $0.844$ ($0.811$).
    
    For particular applications, other symmetries of the assemblages can be explored. These can considerably reduce the computational cost of the procedure. One example is described in the Appendix. It can be used when the measurements in the assemblage preparation have unitary symmetry (as is the case for \textsc{mub}s), and effectively reduces the number of variables we must consider by a factor of $X$ (the number of measurements).

\paragraph{Robustness of high-dimensional steering}
    The optimal value $\eta^*$, as described above, can be interpreted as the distance between the assemblage and the set $\mathcal{A}_k$, measured along the line segment connecting $\mathbf{\sigma}$ to the white noise assemblage. A particularly interesting alternative choice of noise model to be considered is that of all $k$-preparable assemblages. Borrowing ideas from the entanglement and steering robustnesses \cite{vidal-robustness-entanglement,piani-steering-robustness}, we define the $k$-preparability robustness w.r.t.\@ noise model $\mathcal{N}$ as the real number $R_{\mathcal{N}}(\mathbf{\sigma}, k) = t^*$ resulting from the following program, in the limit $N \rightarrow \infty$.
    \begin{subequations}
    \begin{align}
        &\text{min } t \\
        &\text{s.t. } \nonumber \\
        &\quad \frac{\sigma_{a \vert x} + t \pi_{a \vert x}}{1 + t} = \sigma_{a \vert x}^k \label{eq:robustness-mixture} \\
        &\quad \Pi_k^\dag \Big( \ptrb{\substack{(BB^\prime)_i \,:\, i \neq m}}{\Xi_{a \mid x}} \Big) \Pi_k  = \sigma_{a \vert x}^k, \label{eq:robustness-srk-simulation}\\
        &\quad \sum_a \Xi_{a \mid x} = \sum_a \Xi_{a \mid x^\prime}, \,\forall x, x^\prime
    \end{align}
    \end{subequations}
    Since the minimisation ranges over $\Xi_{a, x}$, $t$ and the noise operators $\pi_{a \vert x} \in \mathcal{N}$, this is not a semidefinite program. Nevertheless, when $\mathcal{N}$ is any set defined by a linear matrix inequality (\textsc{lmi}), it can be turned into one, as shown in the Appendix.
    
    We have already shown that the set of Schmidt number $k$ assemblages ($\mathcal{A}_k$) can be approximated by means of a similar hierarchy of semidefinite programs. Thus, any choice $N$ for the hierarchy level will provide an upper bound $R_{\mathcal{N}}^N(\sigma, k) \geq R_{\mathcal{N}}(\sigma, k)$. By taking $\mathcal{N} = \mathcal{A}_k$, we can interpret $R_{\mathcal{A}_k}(\sigma, k)$ as the minimal distance between the assemblage $\sigma$ and the set of Schmidt number $k$ assemblages.

\paragraph{Conclusions}
We have provided a complete set of criteria for the certification of high-dimensional 
steering. The criteria are formulated as a sequence of semidefinite programs, such 
that any kind of high-dimensional steering will be detected at one point of the 
sequence. Moreover, we have demonstrated the practical implementation of our method
and it has turned out that it improves existing results
significantly. 

In the context of known results in quantum steering, our results have further 
consequences. To start with, steering is known to be in one-to-one correspondence 
with measurement incompatibility \cite{quintino2014joint,uola2015one}, a connection 
which was recently extended to high-dimensional steering and $k$-simulatability of 
measurements \cite{jones2022equivalence}, in the sense defined in \cite{ioannou2022simulability}. 
Thus our results also solve the problem of quantifying measurement incompatibility 
in terms of a dimension. On a more practical side, it is also possible to use 
our method to optimise high-dimensional steering witnesses (or its dual to obtain 
them), which might be a fruitful way of constructing tighter witnesses 
than possible with the current approaches \cite{designolle2022robust,qu2022retrieving}.
    
Although our hierarchy for Schmidt number certification can be stated by means 
of semidefinite programs, it leads to a formidable computational problem. Ultimately, 
the size of the problem is dictated by the number of measurements and outcomes, the 
tentative Schmidt number $k$ and the hierarchy level $N$. In particular, the Hilbert 
space dimension grows exponentially with $k$ and $N$. 
A natural further step is thus to find cheaper relaxations or to adapt the formulation 
to specific problems by making use of further symmetries, as done in 
Refs.~\cite{chau-symmetry-measurements,aguilar2018connections,tavakoli2019enabling,ioannou2021noncommutative}.

Another possible extension comes from the key elements in our proposal, which are the lifting 
of the state $\rho$ by means of the ancillary spaces $H_{A^\prime}$ and $H_{B^\prime}$ 
[cf.\@ Eq.~\eqref{eq:lifting-down}], and the reformulation of the symmetric extension 
criterion to the operators $\Omega_{a \vert x}$. In light of recent extensions of the symmetric extensions hierarchy for general cones \cite{aubrun2022monogamy}, it should be possible 
to extend the ideas herein presented to Bell nonlocality and other correlation scenarios.
    
\paragraph{Acknowledgments}
Many thanks to T.\@ Cope, H.\@ C.\@ Nguyen, and R.\@ Uola for useful discussions, 
and to the developers and maintainers of \textsc{Julia} \cite{bezanson2017julia}, \textsc{scs} \cite{ocpb:16} and \textsc{JuMP} \cite{DunningHuchetteLubin2017}. 
    We acknowledge support from the Deutsche Forschungsgemeinschaft (DFG, German     Research Foundation, project numbers 447948357 and 440958198), the Sino-German     Center for Research Promotion (Project M-0294), the ERC (Consolidator Grant 683107/TempoQ), and the German Ministry of Education and Research (Project QuKuK, BMBF Grant No. 16KIS1618K).
    C.G. acknowledges support from the House of Young Talents of the University of Siegen.
    M.P. acknowledges support from the Alexander von Humboldt Foundation.
    R.S. acknowledges financial support by the BMBF project ATIQ and the Quantum Valley Lower Saxony. 
    The \textsc{omni} cluster of the University of Siegen was used for the computations.

\appendix
\section{Appendix A: Proof of Theorem \ref{thm:symmetric-extensions}} 
\label{ap:proof-symmetric-extensions}
    \symmetricextensions*
    The ``if'' condition
    is discussed in the main text, so we now prove the other direction.
    
    \begin{proof}
    The result is an application of \cite[Theorem 1]{aubrun2022monogamy}.
    We will first shortly introduce the necessary notions and state the theorem, then we will show how to apply it to our case.
    
    Let $V$ be a real, finite-dimensional vector space. $C \subset V$ is a proper cone if $C$ is convex, closed, for every $x \in C$ and $\lambda \geq 0$ we have $\lambda x \in C$, $C \cap (-C) = \emptyset$ (i.e., pointed) and $C - C = V$ (generating). Let $V^*$ denote the dual of $V$, then we define the dual cone $C^*$ as the cone of all positive linear functionals, that is, $\psi \in C^* \subset V^*$ if and only if $\psi(x) \geq 0$ for all $x \in C$. We say that $\varphi \in V^*$ is a generalized trace on $C$ if $\varphi \in C^*$ and for all $x \in C$ we have that $\varphi(x) = 0$ only if $x = 0$. It is straightforward to see that positive semidefinite operators form a proper cone and that the usual trace of an operator is a generalized trace as defined above.
    
    Let $V_1, V_2$ be real, finite-dimensional vector spaces and let $C_1 \subset V_1$ and $C_2 \subset V_2$ be proper cones. We define the minimal tensor product as
    \begin{equation}
    C_1 \tmin C_2 = \text{conv}(\{ x_1 \otimes x_2: x_1 \in C_1, x_2 \in C_2 \}),
    \end{equation}
    where ``$\text{conv}$'' denotes the convex hull. The minimal tensor product can be understood as the set of all separable tensors of $C_1$ and $C_2$. Using the notion of dual cones we can also define the maximal tensor product as follows:
    \begin{equation}
    C_1 \tmax C_2 = (C_1^* \tmin C_2^*)^*.
    \end{equation}
    The maximal tensor product can be understood as the largest possible tensors that are positive on separable combinations of the duals. In quantum theory the maximal tensor product coincides with the set of entanglement witnesses.
    
    Let $\varphi_1$ be a generalized trace on $C_1$. Let $x \in C_1 \tmax C_2$ be such that, for any $N$, there is an $x_N \in C_1^{\tmax N} \tmax C_2$ such that:
    \begin{enumerate}
    \item\label{item:symmetric-extension-cond1} $x_N$ is symmetric, i.e., invariant with respect to permutations of the $N$ copies of $C_1$.
    \item\label{item:symmetric-extension-cond2} $(\varphi_1^{N-1} \otimes \text{id}_1 \otimes \text{id}_2)(x_N) = x$, where $\text{id}_1, \text{id}_2$ is the identity map on $V_1, V_2$ respectively. Here, $(\varphi_1^{N-1} \otimes \text{id}_1 \otimes \text{id}_2)(x_N)$ is a generalized partial trace of $x_N$ where we trace out the first $N-1$ copies of $C_1$.
    \end{enumerate}
    Theorem 1 in \cite{aubrun2022monogamy} then states that, under these conditions, we have $x \in C_1 \tmin C_2$. This result can be seen as a generalization of the \textsc{dps} hierarchy \cite{dps-hierarchy} to general cones $C_1$, $C_2$.
    
    We now only need to apply this result to our case. In our case, $C_1$ is the usual cone of positive semidefinite operators generated by quantum states, but $C_2$ is the cone generated by assemblages $\{\sigma_{a|x}\}$ with $a \in [A]$ and $x \in [X]$ where $A$ and $X$ are fixed. Thus, elements of $C_2$ are multiples of assemblages and have the form $\{\lambda \sigma_{a|x}\}$, where $\lambda \geq 0$. It is straightforward to check that $C_2$ is a closed, convex, pointed cone and by letting $V_1 = \text{span}(C_2)$ we get that $C_2$ is also a generating cone. Moreover, we choose the generalized trace on $C_2$ to be given as
    \begin{equation}
    \varphi_1(\{\sigma_{a|x}\}) = \sum_a \tr{ \sigma_{a|x} }.
    \end{equation}
    Notice that the value of the generalized trace does not depend on $x$, since every assemblage is nonsignaling. The dual cone $C_2^*$ is generated by functionals of the form
    \begin{equation}
    \Psi(\{\sigma_{a|x}\}) = \tr{ F \sigma_{a\vert x} }
    \end{equation}
    where $F \geq 0$.
    
    It is straightforward to check that $\Omega_{a|x} \in C_1 \tmax C_2$ simply because $\tr{ \Omega_{a|x} (E \otimes F) } \geq 0$ for all $F \geq 0$ and $E \geq 0$. The result now follows by observing that nonsignaling symmetric extensions of $\Omega_{a|x}$ of order $N$ exactly coincides with an element $x_N \in C_1^{\tmax N} \tmax C_2$ satisfying conditions \ref{item:symmetric-extension-cond1} and \ref{item:symmetric-extension-cond2} above. We thus must have
    \begin{equation}
    \Omega_{a|x} = \sum_i p_i \tau_{a|x}^i \otimes \rho^i
    \end{equation}
    where $\sum_i p_i = 1$, $\tau_{a|x}^i$ is set of assemblages indexed by $i$ and $\rho^i$ are quantum states. This finishes the proof.
    \end{proof}

\section{Appendix B: Measurement symmetries}
\label{ap:symmetries}

    Symmetries in the measurements can be used to further reduce the size of the \textsc{sdp} instances. Suppose, for example, that we prepare the assemblage with measurements on \textsc{mub}s. For simplicity, we focus on prime dimensions $d$, for which the sets of $d$ \textsc{mub}s can be obtained through rotations of the Fourier matrix $F_d$ (as defined in \cite{durt-mubs}, Eq.\@ 5.6), generated by powers of the Hadamard matrices $E_d$ (ibid., Eq.\@ B.1). As such, all of our measurements' effects $\{ M_{a \vert x} \}_{a,x}$ can be expressed from a generating set $\{ M_{a \vert 1} \}_a$ as
    \begin{equation}
        M_{a \vert x} = E_d^{(x-1)} M_{a \vert 1} \big( E_d^{(x-1)} \big)^\dag, \,\forall x \in \{1, \ldots, d\} .
    \end{equation}
    From this, the elements of the assemblage inherit that $\sigma_{a \vert x} = \big( E_d^{(x-1)} \big)^\dag \sigma_{a \vert 1} E_d^{(x-1)}$. Consequently, it is possible to rewrite the constraints as
    \begin{subequations}
    \begin{align}
        &\Pi_k^\dag \Big( \ptrb{\substack{(BB^\prime)_i \\: i \in [N] \setminus m}}{ {\mathcal{E}_d^{(x-1)}}^\dag \Xi_{a \mid 1} \big( \mathcal{E}_d^{(x-1)} \big)} \Big) \Pi_k  = \sigma_{a \vert x} \label{eq:simulation-constraint-with-hadamard}\\
        &\sum_a \Xi_{a \mid 1} = \sum_a {\mathcal{E}_d^{(x - 1)}}^\dag \Xi_{a \mid 1} \big( \mathcal{E}_d^{(x - 1)} \big), \,\forall x, x^\prime ,
    \end{align}
    \end{subequations}
    in which $\mathcal{E}_d = \eye_{A^\prime} \otimes \left( \eye_{B^\prime} \otimes E_d \right)^{\otimes N}$. In this form, the number of variables that we must account for is reduced by a factor of $X$ --- the number of possible measurements.
    
    For particular cases, this strategy may allow us to deal with settings that are at first prohibitively expensive.
    
\section{Appendix C: Steering robustness}
\label{ap:robustness}

 To cast the steering robustness into a semidefinite program, we rewrite Eq.~\eqref{eq:robustness-mixture} as
    \begin{equation}
        \pi_{a \vert x} = \frac{1}{t} \left[ \left( 1 + t \right) \sigma_{a \vert x}^k \right] - \sigma_{a \vert x} \in \mathcal{N} , 
    \end{equation}
    then redefine $\widetilde{\mathcal{N}} \equiv \left\{ t \pi_{a \vert x} \mid t \geq 0, \pi_{a \vert x} \in \mathcal{N} \right\}$ and $\widetilde{\Xi}_{a \vert x} \equiv (1 + t) \Xi_{a \vert x}$. Because $\widetilde{\mathcal{N}}$ is a scalar multiplication of a \textsc{lmi}, it is itself a \textsc{lmi}. And since $\tr{\sum_a \widetilde{\Xi}_{a \vert x}} = 1 + t$ for any choice of $x$, an equivalent expression for $R_{\mathcal{N}}(\sigma, k)$ is
    \nopagebreak\begin{align}
        &\text{min } \text{tr}\Big( \sum_a \widetilde{\Xi}_{a \vert x} \Big) - 1 \\
        &\text{s.t. } \nonumber \\
        &\quad \Pi_k^\dag \Big( \ptrb{\substack{(BB^\prime)_i \,:\, i \in [N] \setminus m}}{\widetilde{\Xi}_{a \mid x}} \Big) \Pi_k - \sigma_{a \vert x} \in \widetilde{\mathcal{N}} \\
        &\quad \sum_a \widetilde{\Xi}_{a \mid x} = \sum_a \widetilde{\Xi}_{a \mid x^\prime}, \,\forall x, x^\prime , 
    \end{align}
    where any $x$ may be chosen in the objective function.

%

\end{document}